\title{\large{\textbf     {
 ON CLUSTER PROPERTIES OF 
 CLASSICAL FERROMAGNETS IN AN EXTERNAL MAGNETIC FIELD
 }}}

\date{}

\documentclass[12pt, leqno]{article}
\usepackage[english]{babel}
\usepackage{amsmath}
\usepackage{amssymb}
\usepackage{amsfonts}
\usepackage{amsthm}
\usepackage[a4paper,left=2.4cm,right=2.4cm,top=2.7cm,bottom=2.7cm]{geometry}
\usepackage{amsfonts, amsmath, amssymb, amsthm, mathrsfs}
\usepackage{graphicx}
\usepackage{color}
\usepackage[sans]{dsfont}

\usepackage{color}

\usepackage{comment}
\usepackage[hang,flushmargin]{footmisc}

\numberwithin{equation}{section}

\newtheorem{thm}{Theorem}
\newtheorem{lem}[thm]{Lemma}

\newtheorem{corollary}[thm]{Corollary}

\newtheorem*{thmbis}{Theorem \ref{T:main}bis}

\theoremstyle{remark}
\newtheorem*{rmk}{Remarks}

\renewcommand{\Re}[0]{\mathrm{Re}}
\renewcommand{\Im}[0]{\mathrm{Im}}

\theoremstyle{definition}

\begin{document}

\maketitle

\begin{center}
\vspace{-1cm}

\textit{Dedicated to David Ruelle and Yasha Sinai, \\
on the occasion of their 80$^{\, th}$ birthday.}\\
\vspace{1.2cm}
J\"urg Fr\"ohlich$^1$ and Pierre-Fran\c cois Rodr\'iguez$^2$ 

\vspace{1.3cm}
Preliminary draft
\end{center}
\vspace{0.5cm}
\begin{abstract}
\centering
\begin{minipage}{0.8\textwidth}
\vspace{0.5cm}
Correlation functions of ferromagnetic spin systems satisfying a Lee-Yang property are studied. It is shown that, for classical systems in a non-vanishing uniform external magnetic field $h$, the connected correlation functions decay exponentially in the distances between the spins, i.e., the \textit{inverse correlation length} (``mass gap''), $m(h)$, is \textit{strictly positive}. Our proof is very short and transparent and is valid for complex values of the  external magnetic field $h$, provided that $\Re \, h \not= 0$. It implies a mean-field lower bound on $m(h)$, as $h \searrow 0$, first established by Lebowitz and Penrose for the Ising model. Our arguments also apply to some quantum spin systems.
\end{minipage}
\end{abstract}

\thispagestyle{empty}

\vspace{1.7cm}

\begin{flushleft}
$^1$Institut f\"ur Theoretische Physik  \hfill December 2015 \\
ETH Z\"urich, HIT K 42.3\\
Wolfgang-Pauli-Strasse 27\\
CH-8093 Z\"urich \\
\texttt{juerg@ethz.ch}

\vspace{1cm}

$^2$Department of Mathematics \\
University of California, Los Angeles \\
520, Portola Plaza, MS 6172\\
Los Angeles, CA 90095 \\
\texttt{rodriguez@math.ucla.edu}
\end{flushleft}

\newpage
\mbox{}
\thispagestyle{empty}
\newpage

\section{Scope of Analysis and Models to be Considered}

In this note, we study a general family of classical ferromagnetic lattice spin systems satisfying the Lee-Yang circle theorem, with the purpose to derive cluster properties of connected correlation (Ursell) functions. The class of models encompasses, for instance, Ising-type models, including ones with continuous one-component spins, and the classical (and quantum) XY-  and Heisenberg models. For such models, we present a new, simple and arguably elegant proof of the claim that the connected two-point function $\langle {\varphi}_0^i \, ; \,  {\varphi}_x^j \rangle^{\Phi}_{ \beta, h}$ (for a precise definition see \eqref{eq:2_pt}, below) decays exponentially in the distance 
$\vert x \vert$, provided the external magnetic field $h$ is in the ``Lee-Yang region'', i.e., whenever $h \in \mathbb{C}$ satisfies $\Re \, h \neq 0$. This result is shown to hold at any inverse temperature $\beta\geq0$ for which the Lee-Yang circle theorem holds. It can be generalized to yield ``tree decay'' of arbitrary connected correlation (Ursell) functions, (i.e., exponential decay in the length of the shortest spanning tree). The range of validity of these results depends on the specific model under consideration and, in particular, on its spin-interaction ``potential'', $\Phi$. A survey of models for which the Lee-Yang theorem holds, along with a fairly exhaustive list of references, can be found in \nolinebreak \cite{FR12}.

Our proof combines a certain representation, derived in \cite{FR12}, of the thermodynamic limits of Ursell functions with periodic boundary conditions, for $h$ in the Lee-Yang region ($\Re \,h \not=0$), with estimates  valid at large values of $\Re \,h$ and established with the help of a cluster expansion and with a straightforward application of the maximum principle in suitably chosen regions of the complex-$h$ half-plane corresponding to $\Re \, h > 0$ (or $\Re \, h < 0$, respectively). The Lee-Yang theorem appears only implicitly in our arguments in that we rely on analyticity results proven in \cite{FR12} that hinge on this theorem but are not presented explicitly in the present paper, anymore. For a limited class of lattice gases, our results were first derived in \cite{LP74}, using subharmonicity arguments.  

The general setup underlying our discussion is as follows. We consider a random field $\varphi = (\varphi_x)_{x \in \mathbb{Z}^d}$ on the $d$-dimensional simple cubic lattice $\mathbb{Z}^{d}$, with $d\geq 2$. The variable $\varphi_{x} = (\varphi^{1}_{x}, ..., \varphi^{N}_{x}) \in \Omega \stackrel{\text{def.}}{=} \mathbb{R}^{N}, N \geq 1,$ describes a classical ``spin'' at site $x \in \mathbb{Z}^{d}$. The energy of a configuration, $\varphi_{\Lambda}$, of spins located in an arbitrary finite subset $\Lambda \subset \subset \mathbb{Z}^d$ is given by the Hamiltonian
\begin{equation}\label{eq:H}
H_{\Lambda}(\varphi) = H^{\Phi}_{\Lambda}(\varphi)- h \sum_{x\in \Lambda} \varphi_x^1, \qquad 
H^{\Phi}_{\Lambda}(\varphi) = \sum _{X \subset \Lambda: \, |X| \geq 2} \Phi(X)(\varphi_X), 
\end{equation}
(with periodic boundary conditions imposed at the boundary of $\Lambda$, in case $\Lambda$ is a rectangle), where $\varphi_X = (\varphi_x)_{x\in X}$, the magnetic field $h \in \mathbb{C}$ is arbitrary, and, for every finite subset $X \subset \mathbb{Z}^{d} $, 
$\Phi(X): \Omega^X\to \mathbb{R}$ is a real-valued, bounded, continuous function representing the interaction energy between the spins in $X$. Besides assuming invariance of $\Phi(\cdot)$ under lattice translations, we have to impose certain conditions on the dependence of the sup-norm of $\Phi(X)$ on $X$; see, e.g., \cite{FR12}, and below. The specification of the model is completed by choosing a finite a-priori measure, $\mu_0$, on the configuration space $\Omega$ of a single spin. We require the following assumptions on $\mu_0$:
\begin{equation} \label{eq:cond_mu_0}
\tag{C1}
\begin{split}
&\text{$\mu_0$ is invariant under rotations of $\Omega$, $\text{supp }\mu_0 \subset \Omega$ is compact, and}\\
&\hspace{ 2.5cm } \hat{\mu}_0(h) \stackrel{\text{def.}}{=}\int e^{h t^1} d\mu_0(\mathbf{t}) \neq 0, \text{ for $\Re \, h \neq 0$}.
\end{split}
\end{equation}
Given that we will make use of the Lee-Yang theorem, our requirement on the Laplace transform of $\mu_0$ is eminently reasonable: it states that the Lee-Yang theorem holds for $\Phi \equiv 0$. The only purpose of the compactness condition is  to avoid troubles with the convergence of certain integrals -- it could be relaxed. The dependence of physical quantities on the choice of $\mu_0$ will usually not be made explicit.

The finite-volume partition function of the system at inverse temperature $\beta > 0$, in the presence of sources described by a bounded continuous function, $f$, of $\varphi = \varphi_{\Lambda}$, is given by
\begin{equation}\label{eq:Z}
Z_{\Lambda,\beta, h}^{\Phi} (f)= \int_{\Omega^{\Lambda}} f(\varphi) \exp[- \beta H^{\Phi}_{\Lambda}(\varphi) ] d\nu_{\Lambda, \beta h} (\varphi),
\end{equation}
where $d\nu_{\Lambda, \beta h} (\varphi) \stackrel{\text{def.}}{=} \prod_{x\in \Lambda} d\nu_{\beta h} (\varphi_x)$, and, for later convenience, we introduce the notation $d\nu_{\beta h} (\mathbf{t})= e^{\beta h t^1} d\mu_0(\mathbf{t}) / \int_{\Omega}  e^{\beta h s^1} d\mu_0(\mathbf{s})$. Thermal averages are denoted by 
$$\langle f \rangle^{\Phi}_{\Lambda, \beta, h} = Z_{\Lambda,\beta, h}^{\Phi}(1)^{-1}Z_{\Lambda,\beta, h}^{\Phi}(f),$$
 with $f$ as above. In order to guarantee convergence of the free energy in the thermodynamic limit $\Lambda \nearrow \mathbb{Z}^d$, the interaction potential $\Phi$ is assumed to be translation invariant (as indicated above) and to satisfy $\sum_{X \ni 0 } |X|^{-1} ||\Phi(X)||_{\infty} < \infty$, where the norm $||\cdot ||_{\infty}$ is the sup-norm on supp $\mu_0^{\otimes X} \subset \Omega^{\times X}$; see, e.g., \cite{Si93}, Chapter 2. Moreover, $\Phi$ is assumed to satisfy a suitable condition of \textit{ferromagnetism} that depends on its specific form. Together with condition \eqref{eq:cond_mu_0} on $\mu_0$, this requirement typically implies a \textit{Lee-Yang property}, which entails that $Z_{\Lambda}^{\Phi}(\beta, h) \neq 0$ whenever $\Re\, h \neq 0$, for suitable values of $\beta$; (for systems with multi-component spins, more specific hypotheses will have to be imposed).

\medskip
\noindent In this paper, $c,c',...,$ denote generic positive constants whose values may change from one place to another. Numbered constants $c_0,c_1,...,$ are defined where they first appear within the text and remain fixed from then on.

\section{Exponential Clustering}

To keep things simple, we temporarily consider models with  pair interactions quadratic in the classical spins. Generalizations to other types of potentials $\Phi$ are discussed in Section \nolinebreak \ref{S:Epilogue}. Thus, we assume that
\begin{equation} \label{eq:cond_J}
\tag{C2} 
\begin{split}
&\Phi(\{ x, y\})(\varphi_x, \varphi_y) = - \sum_{1\leqslant i \leqslant N} J_{xy}^i \varphi_x^i \varphi_y^i, \quad \Phi(X) = 0, \text{ if $|X|\geq 3$},\text{ and }\\
&J_{xy}^1 \geq \sum_{2\leqslant k \leqslant N}|J_{xy}^k| \text{ (``strong'' ferromagnetism),} \\
&J_{xy}^k = 0,\, \forall k, \text{ whenever $|x-y| \geq r$, (finite range),}
\end{split}
\end{equation}
for a given $r \geq1$, and all $x, y \in \mathbb{Z}^d$. The ferromagnetic condition is optimal for $N=2$, but not optimal for $N \geq 3$; (it is incompatible with $O(N)$-symmetry of the model). In the special case where $N=3$ and $\mu_0$ is the uniform measure on the (two-dimensional) unit sphere $S^2$, one can replace it by the more natural constraint $J_{xy}^1 \geq |J_{xy}^2| \vee  |J_{xy}^3|$. This class of models includes the classical Heisenberg model. Under the conditions \eqref{eq:cond_mu_0} and \eqref{eq:cond_J} on $(\mu_0,\Phi)$, which guarantee that a Lee-Yang theorem holds (see, e.g.,  \cite{Ne74}, \cite{DN75} and \cite{LS81}), it was shown in \cite{FR12} that the connected two-point function
\begin{equation}\label{eq:2_pt}
\langle {\varphi}_0^i \, ; \,  {\varphi}_x^j \rangle^{\Phi}_{ \beta, h} \stackrel{\text{def.}}{=} \lim_{\Lambda \nearrow \mathbb{Z}^d} \langle {\varphi}_0^i \, ; \,  {\varphi}_x^j \rangle^{\Phi}_{\Lambda, \beta, h}, \text{ for $1\leq i,j\leq N$, $\beta > 0$, and $\Re \, h \neq 0$},
\end{equation}
with $\langle {\varphi}_0^i \, ; \,  {\varphi}_x^j \rangle^{\Phi}_{\Lambda, \beta, h} =  \langle {\varphi}_0^i  {\varphi}_x^j \rangle^{\Phi}_{\Lambda, \beta, h} -  \langle {\varphi}_0^i \rangle^{\Phi}_{\Lambda, \beta, h} \langle  {\varphi}_x^j \rangle^{\Phi}_{\Lambda, \beta, h}$, is well defined and does not depend on the choice of boundary conditions, so long as the latter do not invalidate the Lee-Yang theorem. This includes in particular free and periodic boundary conditions. The thermodynamic limit, $\Lambda \nearrow \mathbb{Z}^{d}$, in \eqref{eq:2_pt} is understood in the sense of van Hove \cite{Si93}; (readers not familiar with this notion may simply think of taking $\Lambda$ to be a lattice cube centered at $0$ with sides of length $n$, with appropriate boundary conditions imposed at $\partial \Lambda$, and letting $n \to \infty$). 

In our proofs we require one further condition on the a-priori measure $\mu_0$, which guarantees that exponential decay of connected spin-spin correlation functions holds at large values of $|\Re \, h|$, assuming only that $| \Im \, h | \leq const. \times | \Re \, h |$. This will be proven with the help of a ``large-field cluster expansion''; see Section \ref{S:cluster}. Large-field cluster expansions have been developed in \cite{Sp74} in the context of two-dimensional continuum euclidian field theories. For reasons that will become apparent below, it is useful to parametrize the constant in the above inequality as $const. = \tan\alpha$, for $\alpha$ contained in $(0,\frac \pi 2)$. Namely, we assume that,
\begin{equation}\label{eq:cond_Laplace}
\tag{C1'}
\begin{split}
\begin{array}{c}
\text{for all $u_0>0$, there exists
$\tilde{\alpha} = \tilde{\alpha}(u_0) \in (0,\pi/2)$ such that}\\
\displaystyle \min_{u > u_0} \left( \max_{z=u+iv:\vert v \vert \leqslant u\cdot \tan\tilde{\alpha}} \frac{\hat{\mu}_{0}(u)}{\vert\hat{\mu}_{0}(u+iv)\vert} \right) \leq \kappa(\tilde{\alpha}) < \infty,
\end{array}
\end{split}
\end{equation} 
\vspace{0.3cm}\\
for a finite constant $\kappa(\tilde{\alpha})$. We will refer to the (smallest) value of $u > u_0$ achieving the minimum as $\tilde{u}(u_0)$. Condition \eqref{eq:cond_Laplace} can sometimes be strenghtened to
\vspace{0.3cm} 
\begin{equation}\label{eq:cond_Laplace2}
\tag{C1''}
\begin{split}
\begin{array}{c}
\text{the bound in \eqref{eq:cond_Laplace} holds for  } all\\
\displaystyle \tilde{\alpha} \in (0,\pi/2), \text{ and } \kappa =  \sup_{0 <\tilde{\alpha}<  \pi/ 2} \kappa(\tilde{\alpha}) < \infty.
\end{array}
\end{split}
\end{equation} 


\begin{rmk} (i) The denominator in \eqref{eq:cond_Laplace} does not vanish for any measure $\mu_0$ satisfying \eqref{eq:cond_mu_0}. \\
(ii) One can verify condition \eqref{eq:cond_Laplace2} by inspection for some of the most common models. For instance, \eqref{eq:cond_Laplace2} holds for $\mu_0 = \delta_1 + \delta_{-1}$ (Ising spins), and for $\mu_0$ given by the uniform distribution on $S^{N-1}$, for all $N \geq 2$. \\
(iii) Our main result, Eqn. \eqref{eq:main} below, continues to hold under the following weakened version of condition \eqref{eq:cond_Laplace}: Instead of letting $z$ vary along the vertical line segment joining the points
$u(1 \pm i\tan(\tilde{\alpha}))$, with $\tilde{\alpha}= \tilde{\alpha}(u_0)$,  for some $u > u_0$, it suffices that an appropriate bound along \textit{some} smooth curve joining these two endpoints and contained in the set $\{ z' \in \mathbb{C}: \Re \, z' > u_0\}$ hold. 
\end{rmk}
\medskip
We are now ready to state our main result. The inverse correlation length (mass gap), $m (\beta, h)$, defined as
\begin{equation} \label{eq:mass_gap}
m(\beta, h) = - \max_{1\leqslant i,j \leqslant N} \limsup_{|x|\to \infty} \frac{1}{|x|} \log |\langle {\varphi}_0^i \, ; \,  {\varphi}_x^j \rangle^{\Phi}_{ \beta, h}|,
\end{equation}
is a measure for the exponential rate of decay of the two-point function, as $|x| \to \infty$. For pair interactions satisfying assumption \eqref{eq:cond_J} (but not in general(!) -- see Section \ref{S:Epilogue}, below), our arguments turn out to hold at \textit{any} inverse temperature $\beta > 0$. We therefore set $\beta =1$ and omit it from our notation.
\vspace{1cm}
\begin{thm} \label{T:main} (Positivity of the mass gap).

\medskip
\noindent For $(\mu_0, \Phi)$ satisfying conditions \eqref{eq:cond_mu_0}, \eqref{eq:cond_Laplace} and \eqref{eq:cond_J},
\begin{equation} \label{eq:main}
m(h) > 0, \text{ for all real-valued $h \neq 0$}.
\end{equation} 
If, in addition, \eqref{eq:cond_Laplace2} is satisfied, \eqref{eq:main} holds for all values of $h \in \mathbb{C}$ with $\Re \, h \neq 0$.
\end{thm}

\noindent \textit{Remark.} Under the assumptions of Theorem \ref{T:main}, our proof of exponential clustering (tree decay) extends to connected $n$-point correlation functions, for arbitrary $n \geq 2$, and for $h$ in the entire Lee-Yang region; see Section \ref{S:Epilogue}.

\begin{proof} 
Without loss of generality we may assume that $\Re \, h > 0$. 
We introduce two parameters 
\begin{equation}\label{eq:delta_eta}
\delta \in (0,1) \text{ and } \eta \geq 1,
\end{equation}
to be fixed later. The choice of $\eta$ will characterize a region of \textit{large} magnetic fields. We use the notation $\mathbb{H}_+ = \{z \in \mathbb{C}; \, \Re \, z > 0 \}$ 
and define an open domain $\Sigma_{\alpha} \subset \mathbb{H}_+$, depending on an angle $\alpha \in (0,\frac{\pi}{2})$ and the parameters in \eqref{eq:delta_eta}, as follows: 
Let $\mathcal{T}_{\alpha}$ be the interior of the triangle with endpoints $0$ and $p_{\pm} = \eta(1\pm i \tan \alpha)$, and define $\Sigma_{\alpha} = \mathcal{T}_{\alpha}\setminus \overline{D}_\delta(0)$, where the latter refers to the closed disk of radius $\delta$ around the origin.
The boundary of $\Sigma_{\alpha}$ is the union of three subsets, $\partial \Sigma_{\alpha} = \gamma_{\textrm{c}} \cup \gamma_{\textrm{r}} \cup \gamma_{\textrm{v}}$, consisting of a circular part, $\gamma_{\textrm{c}}$, which is an arc of central angle $2 \alpha$ symmetric about the positive real axis, on the circle of radius $\delta$ around $0$, a radial part, $\gamma_{\textrm{r}}$, made up of the two line segments emanating from the endpoints of this arc and ending in $p_+$, resp. $p_-$, and a vertical part, 
$\gamma_{\textrm{v}}$, joining the points $p_+$ and $p_-$. 
The dependence of various quantities on $\delta$ and $\eta$ will usually be implicit.

We fix components $i, j \in \{ 1,\dots, N\}$, as in \eqref{eq:mass_gap}. 
For arbitrary $x \in \mathbb{Z}^d$ and $\varepsilon > 0$, we consider the function $\mathscr{F}_x^{\varepsilon, \alpha}:  \mathbb{H}_+ \to \mathbb{C}$,
\begin{equation}\label{eq:F}
\mathscr{F}_x^{\varepsilon, \alpha}(z) \stackrel{\text{def.}}{=} e^{\varepsilon z
^{\pi / 2 \alpha}|x|} \cdot\langle {\varphi}_0^i \, ; \,  {\varphi}_x^j \rangle^{\Phi}_{z}, \text{ for } z \in \mathbb{H}_+,
\end{equation} 
where, for definiteness, we set $a^b = \exp(b \cdot \text{Log} \, a)$, for $a \in \mathbb{C}\setminus \{ 0\}$, $b \in \mathbb{C}$, with $\text{Log}$ denoting the principal branch of the natural logarithm. We further define 
$\mathscr{F}_x^{\varepsilon}(z) =e^{\varepsilon|x|} \cdot\langle {\varphi}_0^i \, ; \,  {\varphi}_x^j \rangle^{\Phi}_{z}$,  $z \in \mathbb{H}_+$. By Theorem 7 in \cite{FR12} (see also the discussion at the end of Section III therein), the functions $\mathscr{F}_x^{\varepsilon, \alpha}(\cdot)$, $\mathscr{F}_x^{\varepsilon}(\cdot)$  are analytic on $\mathbb{H}_+$, for all 
$\varepsilon > 0$ and $\alpha\in (0,\frac{\pi}{2})$. We propose to derive certain uniform upper bounds for the family $\{ |\mathscr{F}_x^{\varepsilon, \alpha}(\cdot)| \}_{ \, x \in \mathbb{Z}^d }$ on the boundary of the domain 
$\Sigma_{\alpha}$ and then apply the maximum principle in order to extend them to the interior of 
$\Sigma_{\alpha}$. These bounds form the contents of the following two lemmas.

\begin{lem}\label{L:gamma_r_bound} (Estimate along $\gamma_{\textnormal{r}} \cup \gamma_{\textnormal{c}}$; $\alpha \in (0,\frac{\pi}{2})$,  $\delta\in (0,1)$, $\eta \geq 1$). 

\medskip
\noindent There exists $c_0(\alpha, \delta, \eta) \in (0, \infty)$ such that, for all $0 < \varepsilon \leq 1$,
\begin{equation}\label{eq:gamma_r_bound}
 \sup_{x\in \mathbb{Z}^d}  \, \sup_{z \in \gamma_{\textnormal{r}} \cup \gamma_{\textnormal{c}}} |\mathscr{F}_x^{\varepsilon, \alpha}(z)| \leq c_0(\alpha,\delta, \eta)\cdot e^{\varepsilon \delta^{\pi/2\alpha}|x|}.
\end{equation}
\end{lem}

\begin{proof}
First, observe that, if $z \in \gamma_{\textrm{r}}$, then $z = re^{\pm i \alpha}$, for some $r\in [\delta,\eta]$, and the choice $\pm \alpha$ depends on which line segment $z$ belongs to. Regardless, for all $\varepsilon > 0$, $x\in \mathbb{Z}^d$ and $z \in \gamma_{\textrm{r}}$, one has that
$e^{\varepsilon z^{\pi / 2 \alpha}|x|} = e^{\pm i \varepsilon r^{\pi / 2 \alpha}|x|}$, 
i.e., the prefactor appearing in the definition of $\mathscr{F}_x^{\varepsilon, \alpha}$ in \eqref{eq:F} has modulus one. Moreover, when $z \in \gamma_{\textrm{c}}$, writing $z = \delta e^{i\theta}$ for suitable $|\theta|\leq \alpha < \pi/2$, one obtains $|e^{\varepsilon z^{\pi / 2 \alpha}|x|}|
\leq e^{\varepsilon \delta^{\pi/2\alpha}|x|} $. Thus, in order to prove \eqref{eq:gamma_r_bound}, it suffices to show that
\begin{equation}\label{eq:gamma_r_bound1}
\sup_{x\in \mathbb{Z}^d}  \sup_{z \in \gamma_{\textrm{r}}\cup  \gamma_{\textrm{c}}} |\langle {\varphi}_0^i \, ; \,  {\varphi}_x^j \rangle^{\Phi}_{z}| < \infty.
\end{equation}
Note that the left-hand side depends implicitly on $\alpha$, $\delta$ and $\eta$ through the choice of 
$\gamma_{\textrm{r}} \cup \gamma_{\textrm{c}}$. The bound \eqref{eq:gamma_r_bound1} results from a slight extension of the arguments in \cite{FR12}, which we briefly explain. It follows from the proof of Theorem $7$ in \cite{FR12} that
\begin{equation} \label{eq:gamma_r_bound1.1}
\langle {\varphi}_0^i \, ; \,  {\varphi}_x^j \rangle^{\Phi}_{z} = \int_K d^dk e^{i k\cdot x} \widehat{G}_z(k),
\end{equation}
where $K=[0,2\pi)^d$ and the Fourier transform $\widehat{G}_z(k)$ (which depends implicitly on $i$ and $j$) is obtained as
\begin{equation}\label{eq:gamma_r_bound2}
\widehat{G}_z(k) = \frac{\partial f_{\infty}(\xi, k)}{\partial \varepsilon_1 \partial \varepsilon_2}\Big|_{\varepsilon_1 = \varepsilon_2=0}, \quad \text{with} \  \xi = (z,\varepsilon_1, \varepsilon_2) \in \Xi =  \{ \Re \, z > |\varepsilon_1|+|\varepsilon_2 | \} \subset \mathbb{C}^3 
\end{equation}
where $ f_{\infty}$ is a generalized infinite-volume ``free energy''. The precise construction of $f_{\infty}$ in \cite{FR12} is of no concern for the present purposes; but we note that $f_{\infty}$ is obtained by performing a thermodynamic limit, $f_{\infty}= \lim_{L\to \infty} f_{\mathbb{T}_L}$, of suitable finite-volume ``free energies'' $f_{\mathbb{T}_L}$ indexed by $\mathbb{T}_L$, the discrete torus with sides of length $L$. (We recall that, for the arguments used in \cite{FR12} to be valid, we must require periodic boundary conditions.) The functions $f_{\mathbb{T}_L}$ are  analytic on $\Xi$, for every $k \in K$. Careful inspection of the proof of Theorem $7$ in \cite{FR12} reveals that $\sup_{L \geq 1} \sup_{(\xi,k)} f_{\mathbb{T}_L}(\xi,k) < \infty$ for $(\xi,k)$ ranging over arbitrary compact subsets of $ \Xi \times K$, and that the convergence to $f_{\infty}$ is uniform on such sets. Thus, $f_\infty$ is (jointly) analytic in $\xi \in \Xi$, for every $k \in K$, and, moreover, $\sup_{(\xi,k)} f_{\infty}(\xi,k) < \infty$, on compact subsets. Using Cauchy's integral formula for polydiscs, this uniform bound carries over to derivatives of $f_{\infty}$. In particular, in view of \eqref{eq:gamma_r_bound2}, and since $(\gamma_{\textrm{c}} \cup \gamma_{\textrm{r}})\times K$ is compact, it follows that
$$
\sup_{(z,k) \in ( \gamma_{\textrm{c}} \cup \gamma_{\textrm{r}})\times K}  |\widehat{G}_z(k)| < \infty.
$$
Together with \eqref{eq:gamma_r_bound1.1} this immediately yields \eqref{eq:gamma_r_bound1}, and completes the proof of Lemma \nolinebreak \ref{L:gamma_r_bound}.
\end{proof}
\noindent \textit{Remark:} A finer analysis shows that the function $\widehat{G}_z(k)$ is continuous in $(z,k) \in \mathbb{H}_+ \times K$, and that the continuity is uniform on compact subsets. 

\medskip

The estimate \eqref{eq:gamma_r_bound} is complemented by a uniform bound valid for magnetic fields with a sufficiently large real part (as parametrized by $\eta$ and $\alpha$). Recall the definition of $\gamma_{\textrm{v}}=\gamma_{\textrm{v}} (\alpha, \eta) $ above \eqref{eq:F}.

\begin{lem}\label{L:gamma_c_bound} (Large-field estimate).

\medskip
\noindent There exists $c_1 \in [1, \infty) $ with the following properties: \\
 i) Under assumption \eqref{eq:cond_Laplace}, for all $u_0\geq c_1$, with $\tilde{\alpha} = \tilde{\alpha}(u_0) \in (0,\frac \pi 2)$, $\tilde{u}= \tilde{u}(u_0)>u_0$ $($cf. \eqref{eq:cond_Laplace}, and below, for the definition of $\tilde{\alpha}(\cdot)$, $\tilde{u}(\cdot))$ and $\gamma_{\textnormal{v}} = \gamma_{\textnormal{v}}(\tilde{\alpha}, \tilde{u})$, one has that
\begin{equation}\label{eq:large_field1}
\sup_{x \in \mathbb{Z}^d} \sup_{z \in \gamma_{\textnormal{v}}}|\mathscr{F}_x^{m_0}(z)| < c_2,
\end{equation}
for some $m_0>0$ and a possibly large, but finite constant $c_2$.\\
ii) Under the (stronger) assumption \eqref{eq:cond_Laplace2}, the bound in \eqref{eq:large_field1} holds uniformly for all $z$  satisfying $\Re \, z \geq c_1$.
\end{lem}
The proof of \eqref{eq:large_field1} is based on a large-field cluster expansion and is postponed to Section \nolinebreak \ref{S:cluster}. For now, we complete the proof of positivity of the mass gap, see \eqref{eq:main}, using Lemmas \ref{L:gamma_r_bound} and \ref{L:gamma_c_bound}. We assume that, along with conditions (C1) and (C2), (C1') holds. We then pick an arbitrary $h>0$ that will be kept fixed in the following argument.  We must then specify the parameters $\alpha$, $\delta$ and $\eta$ defining the domain $\Sigma_{\alpha}$  (cf. \eqref{eq:delta_eta} and below), as well as the constant $\varepsilon$ in \eqref{eq:F}. We set
\begin{equation}\label{eq:otherparams}
\begin{split}
&\text{      } \delta = \frac{1}{10} \wedge \frac{ h}{2}, \, \eta= \tilde{u}( h \vee c_1)  \text{ and } \alpha= \tilde{\alpha}( h \vee c_1) \\&\text{(with $ \tilde{u}(\cdot),  \tilde{\alpha}(\cdot)$ and $c_1$ as appearing in Lemma \ref{L:gamma_c_bound}, i)).}
\end{split}
\end{equation}
This completely specifies the domain $\Sigma_{\alpha}$. We note that the point $h$ belongs to its interior. Setting 
\begin{equation}\label{eq:def_epsilon}
\begin{split}
&\varepsilon =   m_0 \cdot [ \textstyle \sup_{z \in \gamma_{\textrm{v}}} | \Re \, (z^{\pi / 2 \alpha})| ]^{-1} \wedge 1 \ (>0),
\end{split}
\end{equation}
one finds that
$$
\sup_{z \in \gamma_{\textrm{v}}}|\mathscr{F}_x^{\varepsilon,\alpha}(z)| \stackrel{\eqref{eq:F},\eqref{eq:def_epsilon}}{\leq}   \sup_{z \in \gamma_{\textrm{v}}} |\mathscr{F}_x^{m_0}(z)| \stackrel{\eqref{eq:large_field1}}{<}c_2, \text{ for all } x \in \mathbb{Z}^d.
$$
By Lemma \ref{L:gamma_r_bound}, a bound of the form $c(h)\exp (\varepsilon \delta^{\pi/2\alpha}|x|)$ holds along $\gamma_{\textrm{r}} \cup \gamma_{\textrm{c}}$; (note that $\varepsilon$, $\delta$ and $\alpha$ are all functions of $h$). With Lemma \ref{L:gamma_c_bound}, this implies that all members of the family 
$\{\mathscr{F}_x^{\varepsilon,\alpha}(\cdot) \}_x$ satisfy a bound on all of $\partial \Sigma_{\alpha}$ of the same form. But since each function $\mathscr{F}_x^{\varepsilon,\alpha}(\cdot)$, $x \in \mathbb{Z}^d$, is analytic in $\mathbb{H}_+ \supset \supset \Sigma_{\alpha}$, the maximum principle implies that this bound also holds in the interior of $\Sigma_{\alpha}$, i.e., that $\sup_{x\in \mathbb{Z}^d} \sup_{z \in \Sigma_{\alpha}} \exp (-\varepsilon \delta^{\pi/2\alpha}|x|)\cdot |\mathscr{F}_x^{\varepsilon,\alpha}(z) |\leq c(h).$ In particular, one may choose $z=h \in \Sigma_{\alpha}$ to conclude that 
$\exp (-\varepsilon \delta^{\pi/2\alpha}|x|) \cdot|\mathscr{F}_x^{\varepsilon,\alpha}(h) |\leq c(h)$, for all $x \in \mathbb{Z}^d$. Taking the logarithms of both sides of this inequality, then dividing by $|x|$ and taking $|x|\to \infty$, one deduces that 
\begin{equation}\label{estimate_final1}
- \limsup_{|x|\to \infty} \frac{1}{|x|} \log |\langle {\varphi}_0^i \, ; \,  {\varphi}_x^j \rangle^{\Phi}_{h}| \geq \varepsilon\cdot(h^{\pi/2\alpha} - \delta^{\pi/2\alpha}) \stackrel{\eqref{eq:otherparams}}{\geq} c  \varepsilon h^{\pi/2\alpha},
\end{equation}
which is strictly positive.
Since $i$ and $j$ have only finitely many possible values, cf. \eqref{eq:mass_gap}, one concludes that \eqref{eq:main} holds for all $h>0$.

Next, we consider an arbitrary $h \in \mathbb{H}_+$ and assume that \eqref{eq:cond_Laplace2} holds. We propose to explain how to adapt the above arguments to this situation. We first suppose that $0 < \Re \, h < c_1$, with $c_1$ as in Lemma \ref{L:gamma_c_bound}. We then choose
\begin{equation}\label{eq:final_def_param}
\delta = \frac{1}{10} \wedge \Re (h^{\pi/2\alpha}) \text{   and  }\eta =  c_1,
\end{equation}
for some $\alpha \in (0,\pi/2)$ to be specified shortly. Clearly, there exists an $\alpha_0(h) \in (0, \frac{\pi}{2})$ such that $h \in \Sigma_{\alpha}$, for all $\alpha \geq \alpha_0(h)$. We pick an \textit{arbitrary} such $\alpha$. 
The $\alpha$-dependence of all quantities will henceforth be displayed explicitly. With these choices of parameters, the boundary $\gamma_{\textrm{v}}= \gamma_{\textrm{v}}(\alpha, \eta)$ can be parametrized as $ \frac{c_1}{\cos \theta} e^{i\theta}$, with $|\theta| \leq \alpha$. The definition of $\varepsilon$ in \eqref{eq:def_epsilon} can be recast as
\begin{equation}\label{eq:epsilon_gen}
\varepsilon = 1 \wedge m_0 \cdot  \Big\{ c_1^{\pi/2\alpha} \sup_{\theta: \, |\theta|\leqslant \alpha} \Big| \frac{\cos(\frac{\theta\pi}{2\alpha})}{[\cos(\theta)]^{\pi/2\alpha}}\Big| \Big\}^{-1} \ (= \varepsilon(\alpha)).
\end{equation}
Note that $\varepsilon$ depends on $h$ only through $\alpha$, via the constraint $\alpha \geq \alpha_0(h)$. Repeating previous arguments, one then shows that 
$$\sup_{x\in \mathbb{Z}^d} \sup_{z \in \Sigma_{\alpha}} \exp (-\varepsilon \delta^{\pi/2\alpha}|x|) \cdot |\mathscr{F}_x^{\varepsilon,\alpha}(z) |\leq c_3(\alpha, h).$$ 
Setting $z=h \, (\in \Sigma_{\alpha})$, one obtains the following lower bound, recorded here for later purposes: For all $h$ with $0< \Re \, h < c_1$, and for all $\alpha \geq \alpha_0(h)$,
\begin{equation}\label{eq:posit_explicit}
- \limsup_{|x|\to \infty} \frac{1}{|x|} \log |\langle {\varphi}_0^i \, ; \,  {\varphi}_x^j \rangle^{\Phi}_{h}| \stackrel{\eqref{eq:final_def_param}}{\geq}  c \cdot \varepsilon (\alpha) \cdot \rho(1- \rho^{\frac{\pi}{2\alpha}-1}), \text{ where $\rho =  \Re ( h^{\frac{\pi}{2\alpha}})$},
\end{equation}
with $\varepsilon = \varepsilon (\alpha)$ as in \eqref{eq:epsilon_gen}. Note that the constant $c_3$ does \textit{not} appear in this lower bound, which will be important below. Moreover, if $\Re \, h \geq c_1$, Lemma \ref{L:gamma_c_bound}, ii) immediately implies that the quantity on the left-hand side of \eqref{eq:posit_explicit} is bounded from below by $m_0$. Strict positivity of the mass gap in the entire region $\mathbb{H}_+$ then follows. This completes the proof of Theorem \ref{T:main}, given Lemma \ref{L:gamma_c_bound}.
\end{proof}

The astute reader will have remarked that $\alpha$ is an essentially free parameter in the previous argument. This yields the following

\begin{corollary} \label{C:main2} (Bound on a critical exponent).

\medskip
\noindent Under assumptions \eqref{eq:cond_mu_0}, \eqref{eq:cond_Laplace2} and \eqref{eq:cond_J}, 
\begin{equation} \label{eq:main2}
\inf_{\Re\, h \neq 0} \frac{m(h)}{|\Re \, h|} > 0. 
\end{equation} 
\end{corollary}
\begin{proof}
Let $h \in \mathbb{H}_+$. The case $h \in \mathbb{C} \setminus \overline{ \mathbb{H}}_+$ is handled similarly. Let $\phi_{\alpha}(\theta) = \frac{\cos(\frac{\theta \pi}{2\alpha})}{\cos( \theta)^{{\pi}/{2\alpha}}} = \phi_{\alpha}(-\theta)$, with $\theta \in [-\alpha, \alpha]$ for some $\alpha < \frac{\pi}{2}$. One readily verifies that $\phi_{\alpha}(0)=1$, $\phi_{\alpha}(\alpha)=0$ and $\phi_{\alpha}(\cdot)$ is decreasing on $[0,\alpha]$. Therefore, $\lim_{\alpha \nearrow \frac \pi 2} \sup_{ \theta \in [0, \alpha]} |\phi_{\alpha}(\theta)|  < \infty$. It is then plain from \eqref{eq:epsilon_gen} that $\inf_{\alpha \in (0, \frac \pi 2)} \varepsilon (\alpha) > 0$. Substituting this into \eqref{eq:posit_explicit}, then letting $\alpha \nearrow \frac \pi 2$, the claim \eqref{eq:main2} follows.
\end{proof}


\section{Cluster Expansion} \label{S:cluster}

In this section we sketch a proof of Lemma \ref{L:gamma_c_bound}. We remind the reader that the limit in \eqref{eq:2_pt} exists for \textit{all} boundary conditions not invalidating the Lee-Yang theorem. This class includes, in particular, \textit{free} boundary conditions, which we will impose throughout this section. (But periodic boundary conditions can be used, too.)
\medskip

\noindent \textit{Proof of Lemma \ref{L:gamma_c_bound}}. For convenience, we will use the notation $\langle \cdot \rangle_{0,h}$ to denote an average with respect to $d\nu_{h}(\varphi) = \prod_{x\in \mathbb{Z}^d} d\nu_{h} (\varphi_x)$, (cf. after \eqref{eq:Z}). We begin by setting up a formalism well-suited to describe the cluster expansion. We first note that, for arbitrary $\Lambda \subset \subset \mathbb{Z}^d$ with $\Lambda \supset \{0,x\}$ and $h \in \mathbb{H}_+$,
\begin{equation}\label{eq:cluster1}
\langle {\varphi}_0^i \, ; \,  {\varphi}_x^j \rangle^{\Phi}_{\Lambda, h} = \frac{\partial^2}{\partial s \, \partial t} \log {Z}^{\tau}_{\Lambda, h}((1+s\varphi_0^i)(1+t\varphi_x^{j}))\Big|_{s,t=0}
\end{equation}
where, for a parameter $\tau > 0 $ to be chosen later,
\begin{equation}\label{eq:cluster2}
{Z}^{\tau}_{\Lambda, h}(A) = \Big\langle A \prod_{\{x,y \} \in \mathcal{B}_{\Lambda}} \exp\Big[ \sum_{1\leqslant k \leqslant N} J_{xy}^k(\varphi_x^k \varphi_y^k -\tau^2 \delta_{k,1}) \Big] \Big\rangle_{0,h}.
\end{equation} 
In  \eqref{eq:cluster2}, $A$ is an arbitrary bounded function of the spins in a finite subset of the lattice,
$\delta_{k,1}$ denotes the Kronecker delta, and $\mathcal{B}_{\Lambda}$ is the set of two-element subsets (bonds), $\{x,y\}$,  of $\Lambda$. 
Eq. \eqref{eq:cluster1} holds regardless of the value of $\tau$, because ${Z}^{\tau}_{\Lambda, h}$ can be obtained from ${Z}_{\Lambda, h}= {Z}^{0}_{\Lambda, h}$, defined in \eqref{eq:Z}, by multiplication with a constant independent of $s$ and $t$. We will expand ${Z}^{\tau}_{\Lambda, h}(A)$, for $A= (1+s\varphi_0^i)(1+t\varphi_x^{j})$, in terms of the ``small'' quantities,
\begin{equation}\label{eq:cluster3}
\mu_{X}^{\tau}(\varphi) \stackrel{\text{def.}}{=}
\begin{cases}
s\varphi_0^i, & X=\{0 \}, \\
t\varphi_x^j, & X=\{x \}, \\
e^{\Phi_{\tau}(X)(\varphi)}-1, & |X|=2, 
\end{cases}
\end{equation}
where $\Phi_{\tau}(X) =\sum_{1\leqslant k \leqslant N}\Phi^k_{\tau}(X)$, with $\Phi^k_{\tau}(\varphi_{\{x,y\}})(\varphi):= J_{xy}^k(\varphi_x^k \varphi_y^k -\tau^2 \delta_{k,1})$. Note that $\mu_{X}$ implicitly depends on $s$ and $t$. Introducing the enhanced set of bonds $\widetilde{\mathcal{B}}_{\Lambda} = \mathcal{B}_{\Lambda} \cup \{0\} \cup\{x\}$ -- the reader may want to think of the latter as loop edges attached to $0$ and $x$ -- one shows that $${Z}^{\tau}_{\Lambda, h}((1+s\varphi^{i}_0)(1+t\varphi^{j}_x)) = \langle \prod_{X \in \widetilde{\mathcal{B}}_{\Lambda}} [ 1+ \mu_X^{\tau}(\varphi)] \rangle_{0,h}.$$ 
We introduce a family of polymers $\Gamma:= \{ \zeta \subset \mathbb{Z}^d ; \, 2\leq |\zeta|< \infty\} \cup \{0\} \cup\{x\}$, $\Gamma_{\Lambda}= \{\zeta \in \Gamma : \, \zeta \subset \Lambda\}$, where the ``length'' 
$|\zeta|$ of a polymer $\zeta$ is given by its cardinality. For $\zeta \in \Gamma_{\Lambda}$, let 
$\widetilde{G}_{\zeta}$ be the set of all \textit{connected graphs} with vertex set $\zeta$ and edges belonging to 
$\widetilde{\mathcal{B}}_{\Lambda}$; (loop edges at $0$ and/or $x$ are allowed if $0$ and/or $x$ belong to 
$\zeta$). Each graph $g \in\widetilde{G}_{\zeta}$ is required to contain at least one edge; this is always true if $|\zeta|\geq 2$, and it ensures that the corresponding loop edge is included, in case $\zeta = \{0 \}$ or 
$\{x\}$. Expanding the product over $X \in \widetilde{\mathcal{B}}_{\Lambda}$ in the formula above and factorizing each resulting term into contributions indexed by some polymer $\zeta_i$, with $i=1,...,n$, with the property that $\zeta_{i} \cap \zeta_{j} = \emptyset$, for $i\not= j, i,j=1,...,n$, we obtain that
\begin{equation}\label{eq:cluster4}
{Z}^{\tau}_{\Lambda, h}(A)= 1 + \sum_{n \geqslant 1}\frac{1}{n!} \sum_{\substack{\zeta_1,\dots, \zeta_n \in \Gamma_{\Lambda} \\ \text{pairwise disjoint}}} \prod_{1\leq k \leq n} z_h^{\tau}(\zeta_k),
\end{equation}
with 
\begin{equation} \label{eq:cluster5}
z_h^{\tau}(\zeta)= \Big\langle \sum_{g \in \widetilde{G}_{\zeta}} \prod_{X \in E(g)} \mu_{X}^{\tau}(\varphi) \Big\rangle_{0,h}.
\end{equation}
In these formulas, $A=(1+s\varphi^{i}_0)(1+t\varphi^{j}_x)$ and $E(g)$ denotes the set of edges of the graph $g$, which is a subset of $\widetilde{\mathcal{B}}_{\Lambda}$. We mention that the sum in \eqref{eq:cluster4} is finite, so there are no issues with convergence. It follows that 
${Z}^{\tau}_{\Lambda, h}(A)$ can be regarded as the (grand) partition function of a gas of polymers with hard-core exclusion and activities given by $z_h^{\tau}(\cdot)$. 
The crucial ingredient in this analysis is the following estimate. Recall the definition of the functions 
$\tilde{\alpha}: (0,\infty) \to (0, \frac{\pi}{2})$ and $\tilde{u}: (0,\infty) \to (0, \infty)$, with $\tilde{u}(t) > t$, for $t> 0$; see \eqref{eq:cond_Laplace}.
\begin{lem} (Smallness of activities). \label{L:eq:clust_conv}

\medskip
\noindent For all $\varepsilon \in (0,1)$, there exist constants $\tau(\varepsilon)>0$, $\eta(\varepsilon) > 0$ and $c(\varepsilon)$ with the following properties:\\
i) Under assumptions \eqref{eq:cond_mu_0}, \eqref{eq:cond_Laplace} and \eqref{eq:cond_J}, and for all $u_0 \geq \eta(\varepsilon)$, 
with $\tilde{\alpha}=  \tilde{\alpha}(u_0)$ and $\tilde{u}=  \tilde{u}(u_0)$,
\begin{equation}\label{eq:cluster_conv}
\sup_{y\in \mathbb{Z}^d} \sum_{\substack{\zeta \in \Gamma: \\ y \in \zeta, |\zeta|=n}}|z_h^{\tau(\varepsilon)}(\zeta)| \leq \varepsilon^n,
\end{equation}
for all $n \geq 1$, $h \in \gamma_{\textnormal{v}}(\tilde{\alpha}, \tilde{u})$, and $|s|, |t| < c(\varepsilon)$.

\medskip
\noindent ii) If, in addition, \eqref{eq:cond_Laplace2} holds then the bound in \eqref{eq:cluster_conv} is valid for all $h\in \mathbb{C}$ satisfying $\Re \, h \geq \eta(\varepsilon)$.
\end{lem}

\medskip


\noindent \textit{Proof of Lemma \ref{L:eq:clust_conv}.} We prove $i)$ and then indicate the changes necessary to prove $ii)$. Without loss of generality, we may assume that the function $\kappa(\cdot)$ appearing in \eqref{eq:cond_Laplace} is increasing and continuous on $[0, \tilde{\alpha}(u_0)]$, for given $u_0 >0$, with $\kappa(0)= 1$. Thus, after possibly redefining $\tilde{\alpha}$, we may assume that 
\begin{equation}\label{eq:kappa_unif_bound}
\kappa(\tilde{\alpha}(u_0)) \leq 10, \text{ for all $u_0>0$}.
\end{equation}
Next, we note that, for the measure $d\nu_{h}$ on $\Omega = \mathbb{R}^N$ introduced after \eqref{eq:Z}, one has that 
\begin{equation}\label{eq:cluster_conv_11}
d|\nu|_{h}(\mathbf{t})= M_h \cdot \frac{e^{(\Re \, h)t^1} d\mu_0(\mathbf{t})}{ \int e^{(\Re \, h)u^1} d\mu_0(\mathbf{u})}\equiv M_h \cdot d\mathbb{P}_h (\mathbf{t}),
\end{equation}
for all $h \in \mathbb{H}_+$. In this equation, $\mathbb{P}_h$ is a probability measure on $\mathbb{R}^N$ with compact support; see \eqref{eq:cond_mu_0}. Moreover, by conditions \eqref{eq:cond_mu_0} and \eqref{eq:cond_Laplace}, the factors $M_h= \hat{\mu}_0(\Re\, h)/ | \hat{\mu}_0(h)|$ are well-defined and satisfy $M_{\infty} \equiv \sup_{u_0 >0}\sup_{h \in \gamma_{\textnormal{v}}(\tilde{\alpha}, \tilde{u})} M_h \leq 10 \,$, using \eqref{eq:kappa_unif_bound}; (recall that $\hat{\mu}_0(z)= \int \exp(zu^1) d\mu_0 (\mathbf{u})$). \\ 
Let $\varepsilon \in ( 0,1)$ and $u_0 > 0$ be fixed, and consider the left-hand side of \eqref{eq:cluster_conv}. Suppose first that $n=1$, i.e., $\zeta = \{ 0\}$ or $\{x\}$. In this case, the activity does not depend on $\tau$, and 
$$
|z_h(\zeta)|\stackrel{\eqref{eq:cluster5}, \eqref{eq:cluster3}}{\leq}(|s|\vee |t|) \int d|\nu|_{1,h}(\mathbf{u}) |\mathbf{u}| \stackrel{\eqref{eq:cluster_conv_11}}{\leq}(|s|\vee |t|)M_{\infty} \cdot ||\mu_0||_{\infty}, \text{ for  $h \in \gamma_{\textnormal{v}}(\tilde{\alpha}, \tilde{u})$,}
$$
with $||\mu_0||_{\infty}= \sup\{|{\mathbf{u}}| ; \, \mathbf{u} \in \text{supp}(\mu_0)\}$. From this, \eqref{eq:cluster_conv} follows if $n=1$, provided $|s|, |t| < c\varepsilon$, with $c$ small enough.

Next, suppose that $n\geq 2$, and let $\zeta \in \Gamma$, with $|\zeta|=n$. We write $d|\nu|(\varphi):= \prod_{x\in \mathbb{Z}^d} d|\nu_{h}|(\varphi_x)$, which is a \textit{product} measure on 
$\Omega^{\mathbb{Z}^d}$. We first dispense with possible loop edges appearing in the summation over 
$\widetilde{G}_\zeta$ in \eqref{eq:cluster5}. Let $G_{\zeta}$ denote the set of connected graphs with vertex set $\zeta$ and edges belonging to $\mathcal{B}_{\zeta}$, i.e., without loops edges. With each $g\in  \widetilde{G}_\zeta$ one may associate at most four graphs in $\widetilde{G}_\zeta$ obtained by adding all possible sets of loop edges; adding any such set to $g$ merely introduces an extra factor in 
$ \prod_{X \in E(g)} \mu_{X}^{\tau}({\varphi})$ bounded by $(1\vee ||\mu_0||_{\infty})^2$, $d|\nu|$-a.e, for $|s|,|t|<1$. Thus, for arbitrary $y \in \mathbb{Z}^d$, $\tau > 0$, $|s|,|t|<1$ and $h \in \mathbb{H}_+$,
\begin{equation}\label{eq:cluster_conv_12}
\begin{split}
\sum_{\substack{\zeta \in \Gamma: \\ y \in \zeta, |\zeta|=n}}|z_h^{\tau}(\zeta)| 
&\leq 4(1\vee ||\mu_0||_{\infty})^2 \sum_{\substack{\zeta \in \Gamma: \\ y \in \zeta, |\zeta|=n}}   \sum_{g \in {G}_{\zeta}} \int \Big[ \prod_{X \in E(g)} |\mu_{X}^{\tau}({\varphi})| \Big] d |\nu|({\varphi}).
\end{split}
\end{equation}
For all $g \in {G}_{\zeta}$ and $\tau < ||\mu_0||_{\infty}$, one has the stability estimate 
$$- \sum_{X\in E(g)}  \Phi_{\tau}(X)(\varphi) \leq c \sum_{x\in \zeta}\sum_{y \in \mathbb{Z}^d}  \sum_{1\leq k \leq N}J^k_{xy} \leq c' |\zeta|, \text{ $d|\nu|$-a.e.}
$$ 
Hence, using a tree-graph inequality that can be found, e.g., in \cite{Br86}, Corollary 3.2(a), one obtains the bound  
\begin{equation} \label{eq:cluster_conv_13}
\begin{split}
\sum_{\substack{\zeta \in \Gamma: \\ y \in \zeta, |\zeta|=n}}|z_h^{\tau}(\zeta)| 
&\leq e^{cn} \sum_{\substack{\zeta \in \Gamma: \\ y \in \zeta, |\zeta|=n}}   \sum_{t \in {T}_{\zeta}} \int \Big[ \prod_{X \in E(t)} |\mu_{X}^{\tau}({\varphi})| \Big] \prod_{x\in \zeta} d |\nu|_{1,h}({\varphi}_{x}) \\
&\leq e^{c'n} M_{\infty}^n \sum_{\substack{\zeta \in \Gamma: \\ y \in \zeta, |\zeta|=n}}   \sum_{t \in {T}_{\zeta}} \int \Big[ \prod_{X \in E(t)} |\Phi_{\tau}(X)(\varphi)| \Big] \prod_{x\in \zeta} d \mathbb{P}_{h}({\varphi}_{x}),
\end{split}
\end{equation}
for all $h \in \gamma_{\textnormal{v}}(\tilde{\alpha}, \tilde{u})$, where $T_{\zeta}$ denotes the set of all trees with vertex set $\zeta$ and edges belonging to $\mathcal{B}_{\zeta}$. The second line in \eqref{eq:cluster_conv_13} follows from \eqref{eq:cluster_conv_11} and the elementary inequality $|e^a-1|\leq|a|e^{|a|}$, for $a \in \mathbb{R}$, using that $|E(t)| = n-1$, since $t$ is a tree on $n$ vertices. Note that, because $\Phi_{\tau}$ has finite range (cf. \eqref{eq:cond_J}), any tree $t \in T_{\zeta}$ yielding a non-zero contribution to the right-hand side of \eqref{eq:cluster_conv_13} has the property that $\text{diam}(X) \leq r$, for all $X \in E(t)$. Hence the degree of any vertex in $t$ is bounded by $cr^d<\infty$. It is then easy to see that, given $t \in T_{\zeta}$, one can extract a subset of edges $\hat{E}(t) \subset E(t)$, $|\hat{E}(t)| \geq cn $, with the property that all edges in $\hat{E}(t)$ are ``vertex-disjoint''; (i.e., given any two edges $e\neq e'$ in 
$\hat{E}(t)$, one has that $e\cap e' = \emptyset$). Foregoing factors in \eqref{eq:cluster_conv_13} indexed by $X\in E(t)\setminus \hat{E}(t)$  (at the cost of gaining a factor $e^{cn}$), the fact that all edges in 
$\hat{E}(t)$ are vertex-disjoint implies that the resulting integral factorizes, and one obtains the bound
\begin{equation} \label{eq:cluster_conv_14}
\begin{split}
\sum_{\substack{\zeta \in \Gamma: \\ y \in \zeta, |\zeta|=n}}|z_h^{\tau}(\zeta)| 
&\leq e^{cn} \sum_{\substack{\zeta \in \Gamma: \\ y \in \zeta, |\zeta|=n}}   \Big[ \sum_{t \in \hat{T}_{\zeta}}    \prod_{X \in E(t)} \sum_{1\leq k \leq N} J_X^k \Big] \sup_{\substack{ t\in \hat{T}_{\zeta}, \\ \{x,y\} \in \hat{E}(t)}} \sup_{1\leq k \leq N}  \mathbb{E}_h[ |\varphi_x^k \varphi_y^k - \tau^2 \delta_{k,1}|]^{cn}.
\end{split}
\end{equation}
To estimate the last factor (expectation) on the right side of this expression, one observes that it follows from the definition of $\mathbb{P}_h$ (see \eqref{eq:cluster_conv_11}) by using rotational invariance of $\mu_0$ (cf. \eqref{eq:cond_mu_0}) that $\mathbb{P}_h \stackrel{w}{\rightarrow}\delta_p$, as $\Re \, h \to \infty$, where $\delta_p$ stands for the Dirac mass at \mbox{$p= (||\mu_0||_{\infty}, 0,\dots, 0) \in \mathbb{R}^N$}. In particular, given $\delta > 0$ and setting $\tau(\delta) = ||\mu_0||_{\infty} - \delta$, one can find a large, but finite constant $c_4(\delta)>0$ such that $\mathbb{P}_h(\varphi_0 \in B_{\delta}(p))\geq 1-\delta$, for all $\Re \, h > c_4(\delta)$, (with $B_\delta(p)$ the ball of radius $\delta$ around $p$). Let $\Omega_{x}^{}= \{ \varphi_x \notin B_{\delta}(p)\}$, $x\in \mathbb{Z}^d$. One then finds that
$$
\mathbb{E}_h[ |\varphi_x^k \varphi_y^k - \tau(\delta)^2\delta_{k,1}|] \leq 2 ||\mu_0||_{\infty} \mathbb{P}_h[\Omega_{x} \cup \Omega_{y}]+ \mathbb{E}_h[ |\varphi_x^k \varphi_y^k - \tau(\delta)^2\delta_{k,1}|\cdot1\{ \varphi_x, \varphi_y \in B_{\delta}(p)\}],
$$
for all $\delta>0$ and all $h$ with $\Re \, h > c_4(\delta)$. The first term on the right-hand side is of order $\delta$, by a union bound, and so is the second term, for any $k \in \{1,\dots,N\}$. To see this, one uses that $|\varphi_x^k \varphi_y^k|\leq \delta^2$, for all $k \geq 2$, on the ``event of interest'', and that $|\varphi_x^1 \varphi_y^1 - \tau(\delta)^2|= |(\varphi_x^1-\tau(\delta))(\varphi_y^1 + \tau(\delta))+ \tau(\delta)(\varphi_y^1-\varphi_x^1)|\leq c( |\varphi_x^1-\tau(\delta)| + |\varphi_y^1-\varphi_x^1| ) \leq c \delta$ whenever $\varphi_x, \varphi_y \in B_{\delta}(p)$. Putting things together, one arrives at
\begin{equation} \label{eq:cluster_conv_15}
\sup_{1\leq k \leq N} \mathbb{E}_h[ |\varphi_x^k \varphi_y^k - \tau(\delta)^2\delta_{k,1}|] \leq c \delta, \end{equation}
for all $\delta>0$, $\Re \, h > c_4(\delta)$ and $x,y \in \mathbb{Z}^d$.
In view of \eqref{eq:cluster_conv_14}, it remains to show that the term $ \mathbf{T}:= \sum_{\substack{\zeta \in \Gamma: \\ y \in \zeta, |\zeta|=n}}   \sum_{t \in {T}_{\zeta}}    \prod_{X \in E(t)} J_X $ (with $J_X = \sum_k J_X^k$) grows at most exponentially in $n$. Indeed, since $|{T}_{\zeta}| = n^{n-2}$ when $|\zeta|=n$, one has that 
\begin{equation*}
\mathbf{T}\leq \frac{1}{(n-1)!}\sum_{y_2,\dots y_n} \sum_{t \in {T}_{\{y_1,\dots,y_n\}}}\prod_{X \in E(t)} J_X \leq \frac{n^{n-2}}{(n-1)!} \sup_{t\in T_n} \sum_{y_2,\dots y_n} \prod_{\{ i,j\} \in E(t)} J_{y_iy_j},
\end{equation*}
where $T_{\lbrace y_1,..., y_n \rbrace}$ is the set of trees on $\{ y_1,\dots, y_n\}$, and $y_1=y$. By Stirling's formula, the factor $n^{n-2}/(n-1)!$ grows exponentially in $n$. By induction one shows, with $J = \sum_x J_{0,x}$, that 
$\sum_{y_2,\dots y_n} \prod_{\{ i,j\} \in E(t)} J_{y_iy_j} \leq J^{n-1}$, for all $n \geq 2$. This is obvious when $n=2$, and, in carrying out the induction step, one may assume without loss of generality that $y_n$ is an end point of the tree, i.e., only \textit{one} edge is incident upon it. All in all, this yields the bound $\mathbf{T}\leq e^{cn}$. Finally, substituting this bound and the one in \eqref{eq:cluster_conv_15} into \eqref{eq:cluster_conv_14}, one obtains that 
$$\sum_{\substack{\zeta \in \Gamma: \\ y \in \zeta, |\zeta|=n}}|z_h^{\tau}(\zeta)| \leq e^{cn} \delta^{c'n},$$
 for all $n\geq 2$, $\delta>0$, $y \in \mathbb{Z}^d$ and $h \in \gamma_{\textnormal{v}}(\tilde{\alpha}, \tilde{u})$ satisfying $\Re \, h > c_4(\delta)$. The bound in \eqref{eq:cluster_conv} follows from this estimate and the computation below \eqref{eq:cluster_conv_11} (for $n=1$), upon choosing $\delta= \delta(\varepsilon)$ sufficiently small, for a given $\varepsilon > 0$, and choosing $\eta(\varepsilon)= c_4(\delta(\varepsilon))$.\\
This completes the proof of the first part, $i)$, of Lemma \ref{L:eq:clust_conv}.\\
To prove part $ii)$, one notices that, by virtue of \eqref{eq:cond_Laplace2}, a uniform upper bound on the quantity $M_h$ appearing in \eqref{eq:cluster_conv_11} is obtained, for all $h \in \mathbb{H}_+$ with $\Re \, h > \tilde{u}(u_0=1)\equiv c_5$. 
Thus, adapting the constant $\eta(\varepsilon)= c_4(\delta(\varepsilon)) \vee c_5$, this completes the proof of Lemma \ref{L:eq:clust_conv}.
\hfill $\square$

\medskip
On the basis of the bounds on polymer activities proven in Lemma \ref{L:eq:clust_conv} and using that $\Phi$ has finite range (cf. \eqref{eq:cond_J}, a condition that can actually be relaxed), the proof of Lemma \ref{L:gamma_c_bound} is by now standard; (one sets $c_1 = \eta(\varepsilon = \frac16)$, with $\eta(\cdot)$ as in Lemma \ref{L:eq:clust_conv}; see for instance \cite{Si93}, Theorem V.7.12, and \cite{F90?}, Chapter 2.) \hfill $\square$

\section{Epilogue} \label{S:Epilogue}

We conclude this paper by discussing some extensions of our results. The first one pertains to Theorem \ref{T:main}. Assume that \eqref{eq:cond_mu_0}, \eqref{eq:cond_Laplace} and \eqref{eq:cond_J} hold. Then, as already briefly mentioned, one can prove the following generalization of Theorem \ref{T:main}: Given $n$ points $x_1, \dots, x_n \in \mathbb{Z}^d$, $n \geq 2$, let 
$$
\ell (x_1,\dots, x_n)= \min_{T \in T_n}\sum_{\{ i,j \} \in E(T)}|x_i-x_j|,
$$
denote the length of the ``shortest'' tree connecting these $n$ points, ($|\cdot|$ denoting, e.g., the Manhattan distance on the lattice $\mathbb{Z}^{d}$), where $T_n$ is the set of trees on $\{ 1,\dots, n\}$ and $E(T)$ denotes the set of edges of $T \in T_n$. The (connected) $n$-point function is defined as
\begin{equation*}
\langle \varphi_{x_1}^{i_1}  ;  \, ... \, ;  \varphi_{x_n}^{i_n} \rangle^{\Phi}_{ \beta, h} = \lim_{\Lambda \nearrow \mathbb{Z}^{d}} \Big( \frac{\partial^n}{\partial \varepsilon_1 \cdots \partial \varepsilon_n} \log  \Big\langle  \displaystyle \prod_{1\leqslant k\leqslant n} (1 + \varepsilon_k \varphi_{x_k}^{i_k}) \Big\rangle_{\Lambda, \beta, h}^{\Phi}  \Big) \Big\vert_{\substack{\varepsilon_i =0, \text{ all }i}},
\end{equation*}
for given components $i_k\in \{1,...,N \}$, $1\leq k \leq n$. It is shown in \cite{FR12} that the limit exists, regardless of boundary conditions, so long as they do not invalidate the Lee-Yang theorem, and that,  for all 
$\beta>0$, the function $\langle \varphi_{x_1}  ;  \, ... \, ;  \varphi_{x_n} \rangle^{\Phi}_{ \beta, h}$ is analytic in $h$, for $\mathrm{Re} \ \! h \neq 0$. Setting $\beta=1$ and omitting it from our notation, as above, one has the following generalization of Theorem \ref{T:main}.

\begin{thmbis}\label{T:main_bis}

\medskip
\noindent If $(\mu_0, \Phi)$ satisfy conditions \eqref{eq:cond_mu_0}, \eqref{eq:cond_Laplace} and \eqref{eq:cond_J} then there exists a function $\widetilde{m}(h)>0$ of $h$ such that, for all $n \geq 2$, $x_k \in \mathbb{Z}^d$ and $i_k\in \{1,...,N \}$, $1\leq k \leq n$,
\begin{equation} \label{eq:main_bis}
\langle \varphi_{x_1}^{i_1}  ;  \, ... \, ;  \varphi_{x_n}^{i_n} \rangle^{\Phi}_{ h} \leq c(h) e^{-\tilde{m}(h)\cdot \ell (x_1,\dots, x_n)},
\end{equation} 
for all real-valued $h \neq 0$.
If, in addition, \eqref{eq:cond_Laplace2} is satisfied, \eqref{eq:main_bis} holds for all values of $h \in \mathbb{C}$ with $\Re \, h \neq 0$, and \eqref{eq:main2} holds with $\widetilde{m}(\cdot)$ in place of $m(\cdot)$. 
\end{thmbis}

\bigskip
In essence, the proof of Theorem 1bis is very similar to that of Theorem \ref{T:main} and Corollary \nolinebreak\ref{C:main2}. Only the uniformity of $\widetilde{m}(\cdot)$ in $n$ asserted in Theorem 1bis requires some explanation: One first defines a function similar to the one in \eqref{eq:F}, but with $\langle \varphi_{x_1}^{i_1}  ;  \, ... \, ;  \varphi_{x_n}^{i_n} \rangle^{\Phi}_{ h}$ replacing the two-point function. The results in \cite{FR12}  guarantee that this function is analytic in $h$ on $\mathbb{H}_+$. Using a slight generalization of Lemma \ref{L:eq:clust_conv} and a standard cluster expansion, one shows that Lemma \ref{L:gamma_c_bound} continues to hold for all $n$-point functions, with a constant $c_2$ that might depend on $n$ (to which $\tilde{m}(\cdot)$ is not sensitive), but with a constant $m_0$ that is uniform in $n$; see, e.g., \cite{Si93}, Theorem V.7.13. With regard to the radial estimate of Lemma \ref{L:gamma_r_bound}, a dependence of our bounds in \eqref{eq:gamma_r_bound} on $n$ may only occur through the constant $c_0$, which is, however, inconsequential; cf. \eqref{estimate_final1} and \eqref{eq:posit_explicit}. Theorem 1bis then follows by arguments very similar to those in Sect. 2.

In order to keep the exposition clear, we have considered the specific (class of) interactions \eqref{eq:cond_J}, but our methods are in fact quite robust, and apply to all models $(\mu_0, \Phi)$ discussed in \cite{FR12}. For instance, in the scalar case $N=1$ and for the spin-$\frac12$ reference measure, a Lee-Yang theorem is known to be true at sufficiently small temperatures if one adds a suitably small quartic interaction, see \cite{LR11}, \cite{Ru10}. Theorem \ref{T:main}, Corollary \ref{C:main2} and Theorem \nolinebreak 1bis still hold in this case, and only require adapting the arguments of Lemma \ref{L:eq:clust_conv} in the expansion at large magnetic fields.

Moreover, our results can be extended to certain quantum spin systems satisfying a Lee-Yang theorem, in particular to the quantum XY- and Heisenberg models. For these models, a large-field cluster expansion can be set up and shown to converge. Finally, our ideas could be applied to the analysis of N-component Euclidean $\lambda|\mathbf{\phi}|^4_d$ -field theories for $N=1,2,3$, with $\mathbf{\phi} = (\phi^1,\dots, \phi^N)$, in $d = 2, 3$ space-time dimensions and in the presence of a ``magnetic field'' term. For this purpose, one has to combine results sketched in \cite{FR12} with an expansion akin to the one in \cite{Sp74}, due to T. Spencer. A more detailed discussion of these matters goes beyond the scope of the present note.

\bibliography{rodriguez}
\bibliographystyle{plain}

\end{document}